%% file: main_color_v2.tex
\documentclass[journal,twoside,web]{ieeecolor}

\usepackage{generic}
\usepackage{cite}
\usepackage{amsmath,amssymb,amsfonts}
\usepackage{graphicx}
\usepackage{textcomp}

\usepackage{graphicx}
\usepackage{epsfig} 
\usepackage{cite}
\usepackage{lcsys}
\usepackage[hidelinks]{hyperref}

\usepackage{amsmath}
\usepackage{amsthm}
\usepackage{mathrsfs}
\usepackage{amssymb}
\usepackage{color}
\usepackage{cite}
\usepackage{float}
\usepackage{graphicx}
\usepackage{booktabs}
\usepackage{longtable,tabularx}
\usepackage{nicefrac}
\usepackage{mathtools}
\usepackage{bm}
\usepackage{multirow}
\usepackage{apptools}
\AtAppendix{\counterwithin{lma}{subsection}}
\usepackage{siunitx}
\usepackage{soul}
\usepackage{algorithm}
\usepackage[noend]{algpseudocode}

\input{shortcuts.tex}

\newcommand{\blue}[1]{ #1}

\newcommand{\red}[1]{\iffalse #1 \fi}

\DeclareMathOperator*{\argmin}{arg\,min}

\newtheorem{lemalt}{Lemma}[lma]

\newenvironment{lema}[1]{
  \renewcommand\thelemalt{#1}
  \lemalt
}{\endlemalt}

\newcounter{example}
\newenvironment{example}[1][]{\refstepcounter{example}\par\medskip
   \noindent \textbf{\indent Example~\theexample. #1} \rmfamily}{\medskip}

\allowdisplaybreaks

\begin{document}

\def\BibTeX{{\rm B\kern-.05em{\sc i\kern-.025em b}\kern-.08em
    T\kern-.1667em\lower.7ex\hbox{E}\kern-.125emX}}
\markboth{\journalname, VOL. XX, NO. XX, XXXX 2017}
{Author \MakeLowercase{\textit{et al.}}: Preparation of Papers for IEEE Control Systems Letters (August 2022)}

\title{Efficient Batch and Recursive Least Squares for Matrix Parameter Estimation} 

\author{Brian Lai and Dennis S. Bernstein
\thanks{
This work was supported by the NSF Graduate Research Fellowship under Grant No. DGE 1841052. 
\textit{(Corresponding author: Brian Lai.)} \\
The authors are with the Department of Aerospace Engineering, University of Michigan, Ann Arbor, MI 48109 USA (e-mail: brianlai@umich.edu; dsbaero@umich.edu).
}
}

\maketitle
\thispagestyle{empty}
%
\begin{abstract}
    Traditionally, batch least squares (BLS) and recursive least squares (RLS) are used for identification of a vector of parameters which form a linear model. 
    In some situations however, it is of interest to identify parameters in a matrix structure. 
    In this case, a common approach is to transform the problem into standard vector form using the vectorization (vec) operator and the Kronecker product, known as vec-permutation.
    However, the use of the Kronecker product introduces extraneous zero terms in the regressor, resulting in unnecessary additional computational and space requirements. 
    This work derives matrix BLS and RLS formulations which, under mild assumptions, minimize the same cost as the vec-permutation approach.
    This new approach requires less computational complexity and space complexity than vec-permutation in both BLS and RLS identification.
    It is also shown that persistent excitation guarantees convergence to the true matrix parameters. 
    This method can used to improve computation time in the online identification of multiple-input, multiple-output systems for indirect adaptive model predictive control.
\end{abstract}


\section{Introduction}
Least squares based identification methods are foundational to systems and control theory, particularly identification, signal processing, and adaptive control \cite{aastrom1995adaptive,ljung1983theory}.
Batch least squares (BLS) and recursive least squares (RLS) are traditionally used to identify a vector of parameters in a linear measurement process \cite{ljung1983theory,islam2019recursive}.
However, it may be of interest to identify parameters in a matrix structure, for example, in adaptive control of multiple-input, multiple-output (MIMO) systems \cite{nguyen2021predictive,islam2019recursive}.
One approach is to use \textit{vec-permutation} \cite{henderson1981vec}, a method which rewrites the linear measurement process such that the columns of the parameters to be identified are stacked into a vector. 
This is accomplished using the the vectorization operator and Kronecker product, and is a straightforward solution for various situations
\cite{islam2021data,nguyen2021predictive,zhu2021recursive,ding2013coupled,wang2018recursive,mohseni2022predictive,farahmandi2024predictive,ma2019recursive}.

A significant drawback, however, is that the vec-permutation method increases the dimension of the linear measurement process by using the Kronecker product, introducing extraneous zero terms in the regressor (e.g. equation (15) of \cite{nguyen2021predictive}).
This results in increased computational cost and storage requirements.
\blue{Another approach is to apply standard least squares methods to separately identify the columns of the matrix of parameters \cite[p. 102]{ljung1983theory} but this method does not address what cost function is being minimized or the relationship to the vec-permutation approach.}
Other related methods including square root filtering \cite{peterka1975square}, multiinnovations \cite{ding2009multiinnovation}, and gradient-based methods \cite{bamieh2002identification} also do not address whether a least squares cost function is globally minimized or the relationship to standard least squares methods.

This work derives a batch and recursive least squares algorithm for identification of matrix parameters which, under the assumption of independent residual error and parameter column weighting, minimizes the same cost function used in the vec-permutation approach. 
This method provides an $\mathcal{O}(m^3)$ times improvement in computational complexity and an $\mathcal{O}(m^2)$ times improvement in storage requirements over vec-permutation, where $m \ge 1$ is the number of columns of the identified parameter matrix.
We also show how persistent excitation guarantees convergence of the identified matrix parameters to true matrix parameters, which extends established results for identification of vector parameters \cite{bruce2021necessary,lai2021regularization}.
\blue{However, we show this improvement in computational complexity may come at the cost of performance if the columns of measurement noise are highly correlated.}
Finally, we show how this method can be used to significantly reduce computation time spent on online identification in predictive cost adaptive control (PCAC) \cite{nguyen2021predictive}.

%
%
%

\section{Vec-permutation Least Squares}

Consider a measurement process of the form\footnote{Note that since the measurement, regressor, and parameters are all matrices, the results of this work can be easily extended to measurements processes of the form $y_k = \theta \phi_k$ by rewriting as $y_k^\rmT = \phi_k^\rmT \theta^\rmT$ and identifying parameters $\theta^\rmT$. For brevity, we leave the details to the reader.}
\begin{align}
    y_k = \phi_k \theta, \label{eqn: yk = phik theta}
\end{align}
where $k = 0,1,2,\hdots$ is the time step, $y_k \in \BBR^{p \times m}$ is the measurement at step $k$, $\phi_k \in \BBR^{p \times n}$ is the regressor at step $k$, and $\theta \in \BBR^{n \times m}$ is a matrix of unknown parameters.
Parameters $\theta$ can be identified by minimizing the least squares cost function $J_k \colon \BBR^{n \times m} \rightarrow \BBR$, defined as
\begin{align}
\label{eqn: RLS cost, no assumptions}
    J_k(\hat{\theta}) = & \sum_{i=0}^k \vecc(y_i - \phi_i \hat{\theta})^\rmT \bar{\Gamma}_i \vecc(y_i - \phi_i \hat{\theta}) \nonumber \\
    & + \vecc(\hat{\theta} - \theta_0)^\rmT \bar{R} \vecc(\hat{\theta} - \theta_0),
\end{align}
where $\vecc(\cdot)$ is the column stacking operator, positive-definite \blue{(and thus, by definition,
symmetric)} $\bar{R} \in \BBR^{mn \times mn}$ is the regularization matrix, $\theta_0 \in \BBR^{n \times m}$ is an initial estimate of $\theta$, and, for all $k \ge 0$, positive-definite $\bar{\Gamma}_k \in \BBR^{mp \times mp}$ is the weighting matrix.
\blue{If no estimate of the parameters $\theta$ is known, common choices in practice for $\theta_0$ are $\theta_0 = 0$ or randomly sampled values.}
\blue{Note that while the regularization term $\vecc(\hat{\theta} - \theta_0)^\rmT \bar{R} \vecc(\hat{\theta} - \theta_0)$ results in a biased estimate, this term reduces variance and guarantees that the least squares cost \eqref{eqn: RLS cost, no assumptions} has a unique global minimizer which is useful when few data have been collected \cite{lai2021regularization}\cite[sec. 7.3]{hastie2009elements}.}

Using vec-permutation \cite{henderson1981vec}, \eqref{eqn: yk = phik theta} can be rewritten as
\begin{align}
    \bar{y}_k = \bar{\phi}_k \bar{\theta},
\end{align}
where $\bar{y}_k \in \BBR^{mp}$, $\bar{\phi}_k \in \BBR^{mp \times mn}$, and $\bar{\theta} \in \BBR^{mn}$ are defined
\begin{align}
    \bar{y}_k &\triangleq \vecc(y_k), \label{eqn: ybark defn} \\
    \bar{\phi}_k &\triangleq I_m \otimes \phi_k, \label{eqn: phibark defn} \\
    \bar{\theta} &\triangleq \vecc(\theta),
\end{align}
and where $\otimes$ is the Kronecker product.
Note, for all $k \ge 0$ and $\hat{\theta} \in \BBR^{n \times m}$ that $\vecc(y_k - \phi_k \hat{\theta}) = \bar{y}_k - \bar{\phi}_k \vecc(\hat{\theta})$ and $\vecc(\hat{\theta} - \theta_0) = \vecc(\hat{\theta}) - \bar{\theta}_0$,
where $\bar{\theta}_0\in \BBR^{mn}$ is defined as
\begin{align}
    \bar{\theta}_0 \triangleq \vecc(\theta_0).
\end{align}
It then follows that, for all $k \ge 0$, the cost function $J_k$, given in \eqref{eqn: RLS cost, no assumptions}, can be rewritten as
\begin{align}
\label{eqn: RLS cost, vec permutation}
    J_k(\hat{\theta}) = & \sum_{i=0}^k (\bar{y}_i - \bar{\phi}_i \vecc(\hat{\theta}))^\rmT \bar{\Gamma}_i (\bar{y}_i - \bar{\phi}_i \vecc(\hat{\theta})) \nonumber \\
    & + (\vecc(\hat{\theta}) - \bar{\theta}_0)^\rmT \bar{R} (\vecc(\hat{\theta}) - \bar{\theta}_0).
\end{align}
Propositions \ref{prop: batch vec permutation} gives the vec-permutation approach to minimize cost function $J_k$.


\blue{
\begin{prop}
\label{prop: batch vec permutation}
    For all $k \ge 0$, let $\phi_k \in \BBR^{p \times n}$, let $y_k \in \BBR^{p \times m}$, and let $\bar{\Gamma}_k \in \BBR^{mp \times mp}$ be positive definite. 
    Furthermore, let $\bar{\theta}_0 \in \BBR^{n \times m}$  and let $\bar{R} \in \BBR^{mn \times mn}$ be positive definite. 
    Then, for all $k \ge 0$, $J_k \colon \BBR^n \rightarrow \BBR$, defined in \eqref{eqn: RLS cost, no assumptions}, has a unique minimizer, whose vectorization is denoted as $\bar{\theta}_{k+1} \triangleq \vecc ( \argmin_{\hat{\theta} \in \BBR^{n \times m}} J_k(\hat{\theta}) )$, which is given by
    \begin{align}
    \label{eqn: batch vec permutation}
        \bar{\theta}_{k+1} = \bar{A}_k^{-1} \bar{b}_k,
    \end{align}
    where
    \begin{align}
        \bar{A}_k &\triangleq \bar{R} + \sum_{i=0}^k \bar{\phi}_i^\rmT \bar{\Gamma}_i \bar{\phi}_i, 
        \label{eqn: batch vec permutation A}
        \\
        \bar{b}_k &\triangleq \bar{R} \vecc(\theta_0) + \sum_{i=0}^k \bar{\phi}_i^\rmT \bar{\Gamma}_i \bar{y}_i,
        \label{eqn: batch vec permutation b}
    \end{align}
    and where, for all $k \ge 0$, $\bar{y}_k \in \BBR^{mp}$ and $\bar{\phi}_k \in \BBR^{mp \times mn}$ are defined in \eqref{eqn: ybark defn} and \eqref{eqn: phibark defn}, respectively.
    Furthermore, for all $k \ge 0$, $\bar{\theta}_{k+1}$ is given recursively by
    \begin{align}
        \bar{P}_{k+1}^{-1} &= \bar{P}_k^{-1} + \bar{\phi}_k^\rmT \bar{\Gamma}_k \bar{\phi}_k, 
        \label{eqn: recursive vec permutation Pkinv update}
        \\
        \bar{\theta}_{k+1} &= \bar{\theta}_k + P_{k+1} \bar{\phi}_k^\rmT \bar{\Gamma}_k (\bar{y}_k - \bar{\phi}_k \bar{\theta}_k).
        \label{eqn: recursive vec permutation thetak update}
    \end{align}
    where $\bar{P}_0 \triangleq \bar{R}^{-1}$, and, for all $k \ge 0$, $\bar{P}_k \in \BBR^{mn \times mn}$ is positive definite, and hence nonsingular.
    Moreover, for all $k \ge 0$, $\bar{P}_k$ can be expressed recursively as
    \begin{align}
        \bar{P}_{k+1} &= \bar{P}_k - \bar{P}_k \bar{\phi}_k^\rmT(\bar{\Gamma}_k^{-1} + \bar{\phi}_k \bar{P}_k \bar{\phi}_k^\rmT)^{-1} \bar{\phi}_k \bar{P}_k.
        \label{eqn: recursive vec permutation Pk update}
    \end{align}
\end{prop}

\begin{proof}
    Since $J_k \colon \BBR^n \rightarrow \BBR$, defined in \eqref{eqn: RLS cost, no assumptions}, is a standard least squares cost with vector parameters, this result follows directly from \cite{islam2019recursive}.
\end{proof}
}
\blue{Equations \eqref{eqn: batch vec permutation} through \eqref{eqn: recursive vec permutation thetak update} give the batch least squares solution using vec-permutation while \eqref{eqn: recursive vec permutation Pkinv update} through \eqref{eqn: recursive vec permutation Pk update} give the recursive least squares solution using vec-permutation.}
An inefficiency with this method is that the Kronecker product in \eqref{eqn: phibark defn} introduces extraneous zero terms in $\bar{\phi}_k$ when $m > 1$, resulting in a sparse and higher dimensional regressor matrix.
However, the results of Proposition \ref{prop: batch vec permutation} cannot be simplified since the regularization matrix $\bar{R}$ and weighting matrices $\bar{\Gamma}_k$ are not necessarily sparse.
\red{The computational complexities of BLS and RLS with vec-permutation are shown in Tables \ref{table: batch} and \ref{table: recursive} respectively.}

\section{Column-by-Column Least Squares}

\blue{To simplify the vec-permutation approach}, we make the assumption that there exist positive-definite $R_1,\hdots,R_m \in \BBR^{n \times n}$ such that $\bar{R}$ is block diagonal of the form
\begin{align}
    \bar{R} = 
    \diag(R_1,\hdots,R_m).
    \label{eqn: Rbar block diag}
\end{align}
Furthermore, we assume that, for all $k \ge 0$, there exist positive-definite $\Gamma_{1,k},\hdots,\Gamma_{m,k} \in \BBR^{p \times p}$ such that $\bar{\Gamma}_k$ is block diagonal of the form
\begin{align}
    \bar{\Gamma}_k = 
    \diag(\Gamma_{1,k},\hdots,\Gamma_{m,k}).
    \label{eqn: Gammabar block diag}
\end{align}
This corresponds to independent weighting of the columns of the residual error, $y_i - \phi_i \hat{\theta}$, $i = 0,\hdots,k$, and of the regularization term, $\hat{\theta} - \theta_0$, in \eqref{eqn: RLS cost, no assumptions}.
Then, for all $k \ge 0$ and $\hat{\theta} \in \BBR^{n \times m}$, \eqref{eqn: RLS cost, no assumptions} can be rewritten as
\begin{align}
    \label{eqn: Jk independent column weighting}
    J_k(\hat{\theta}) = \sum_{j=0}^m J_{j,k} (\hat{\theta}_j),
\end{align}
where, for all $j = 1,\hdots,m$, $J_{j,k} \colon \BBR^n \rightarrow \BBR$ is defined as
\begin{align}
        J_{j,k}(\hat{\theta}_j) = & \sum_{i=0}^k (y_{j,i}-\phi_i \hat{\theta}_j)^\rmT \Gamma_{j,i} (y_{j,i}-\phi_i \hat{\theta}_j) \nonumber \\
    & + (\hat{\theta}_j - \theta_{j,0})^\rmT R_j (\hat{\theta}_j - \theta_{j,0}).
\end{align}
where the vectors $y_{1,k},\hdots,y_{m,k} \in \BBR^{p}$, $\theta_{1,0},\hdots,\theta_{m,0} \in \BBR^n$, and $\hat{\theta}_1,\hdots,\hat{\theta}_m \in \BBR^n$ are the $m$ columns of $y_k$, $\theta_0$, and $\hat{\theta}$, respectively. In particular,
\begin{align}
    y_k &\triangleq \begin{bmatrix} y_{1,k} & \cdots & y_{m,k} \end{bmatrix}, \label{eqn: y_j,k defn} \\
    \theta_0 &\triangleq \begin{bmatrix} \theta_{1,0} & \cdots & \theta_{m,0} \end{bmatrix}, \label{eqn: theta_j,0 defn} \\
    \hat{\theta} &\triangleq  \begin{bmatrix} \hat{\theta}_1 & \cdots & \hat{\theta}_m \end{bmatrix}.
    \label{eqn: thetahat_j defn}
\end{align}
\red{The following Lemma shows that minimizing the cost function $J_k$, given by \eqref{eqn: Jk independent column weighting}, can be done by separately minimizing $J_{j,k}$ for all $j = 1,\hdots,m$.}
\red{
\begin{lem}
\label{lem: Jk min partition}
For all $k \ge 0$, let $\phi_k \in \BBR^{p \times n}$, let $y_k \in \BBR^{p \times m}$, and let ${\Gamma}_{1,k},\hdots,{\Gamma}_{m,k} \in \BBR^{p \times p}$ be positive definite. 
    Furthermore, let $\bar{\theta}_0 \in \BBR^{n \times m}$  and let ${R}_{1},\hdots,R_{m} \in \BBR^{n \times n}$ be positive definite. 
    Then, for all $k \ge 0$, $J_k \colon \BBR^n \rightarrow \BBR$, defined in \eqref{eqn: Jk independent column weighting}, has a unique minimizer given by
\begin{align*}
        \argmin_{\hat{\theta} \in \BBR^{n \times m}} J_k(\hat{\theta}) 
        & = \begin{bmatrix}
           \argmin\limits_{\hat{\theta}_1 \in \BBR^{n}}J_{1,k}(\hat{\theta}_1)  & \hspace{-1.5pt}
           \hdots & 
           \hspace{-1.5pt} \argmin\limits_{\hat{\theta}_m \in \BBR^{n}}J_{m,k}(\hat{\theta}_m) 
        \end{bmatrix}
\end{align*}
\end{lem}

\begin{proof}
    Note that, for all $j = 1,\hdots,m$, since $J_{j,k}$ is a function of only $\hat{\theta}_j$, it follows from Lemma \ref{lem: least squares} that $J_{j,k}$ has a unique minimizer. 
    Then, since $J_k(\hat{\theta}) = \sum_{j=0}^m J_{j,k} (\hat{\theta}_j)$, it follows from \eqref{eqn: thetahat_j defn} that the Lemma holds.
\end{proof}
}

\blue{Propositions \ref{prop: batch independent} show that, under assumptions \eqref{eqn: Rbar block diag} and \eqref{eqn: Gammabar block diag}, the cost function $J_k$, given by \eqref{eqn: Jk independent column weighting}, can be minimized by separately updating the columns of the parameter estimate. We call this the \textit{column-by-column} approach.}

\blue{
\begin{prop}
\label{prop: batch independent}
   For all $k \ge 0$, let $\phi_k \in \BBR^{p \times n}$, let $y_k \in \BBR^{p \times m}$, and let ${\Gamma}_{1,k},\hdots,{\Gamma}_{m,k} \in \BBR^{p \times p}$ be positive definite. 
    Furthermore, let $\bar{\theta}_0 \in \BBR^{n \times m}$  and let ${R}_{1},\hdots,R_{m} \in \BBR^{n \times n}$ be positive definite. 
    Then, for all $k \ge 0$, $J_k \colon \BBR^n \rightarrow \BBR$, defined in \eqref{eqn: Jk independent column weighting}, has a unique minimizer, whose columns are denoted as 
    \begin{align}
    \label{eqn: independent minimizer partition}
        \begin{bmatrix}
            \theta_{1,k+1} & \cdots & \theta_{m,k+1}
        \end{bmatrix}
        \triangleq 
        \argmin_{\hat{\theta} \in \BBR^{n \times m}} J_k(\hat{\theta}),
    \end{align}
    which, for all $j = 1,\hdots,m$, are given by
    \begin{align}
    \label{eqn: batch independent}
        \theta_{j,k+1} = A_{j,k}^{-1} b_{j,k},
    \end{align}
    where
    \begin{align}
        A_{j,k} &\triangleq R_j + \sum_{i=0}^k \phi_i^\rmT \Gamma_{j,i} \phi_i, 
        \label{eqn: batch independent A}
        \\
        b_{j,k} &\triangleq R_j \theta_{j,0} + \sum_{i=0}^k \phi_i^\rmT \Gamma_{j,i} y_{j,i},
        \label{eqn: batch independent b}
    \end{align}
    and where, for all $k \ge 0$, $y_{j,k} \in \BBR^p$ is defined in \eqref{eqn: y_j,k defn} and $\theta_{j,0} \in \BBR^n$ is defined \eqref{eqn: theta_j,0 defn}.
    Furthermore, for all $k \ge 0$ and $j = 1,\hdots,m$, $\theta_{j,k+1} \in \BBR^n$ is given recursively by
    \begin{align}
        P_{j,k+1}^{-1} &= P_{j,k}^{-1} + \phi_{k}^\rmT \Gamma_{j,k} \phi_{k}, 
        \label{eqn: recursive independent Pkinv}
        \\
        \theta_{j,k+1} &= \theta_{j,k} + P_{j,k+1} \phi_{k}^\rmT \Gamma_{j,k} (y_{j,k} - \phi_k \theta_{j,k}).
        \label{eqn: recursive independent thetak}
    \end{align}
    where $P_{j,0} \triangleq R_j^{-1}$, and, for all $k \ge 0$ and $j = 1,\hdots,m$, $P_{j,k} \in \BBR^{n \times n}$ is positive definite, and hence nonsingular. 
    Moreover, for all $k \ge 0$ and $j=1,\hdots,m$, $P_{j,k}$ can be expressed recursively as
    \begin{align}
        P_{j,k+1} &= P_{j,k} - P_{j,k} \phi_{k}^\rmT(\Gamma_{j,k}^{-1} + \phi_{k} P_{j,k} \phi_{k}^\rmT)^{-1} \phi_{k} P_{j,k}.
        \label{eqn: recursive independent Pk}
    \end{align}
\end{prop}

\begin{proof}
    Note that, for all $j = 1,\hdots,m$, $J_{j,k}$ is a function of only $\hat{\theta}_j$ and is a standard least squares cost with vector parameters. 
    It then follows from \cite{islam2019recursive} that, for all $j=1,\hdots,m$, $J_{j,k}$ has a unique minimizer, denoted as $\theta_{j,k+1} \triangleq \argmin_{\hat{\theta}_j \in \BBR^{n}}J_{j,N}(\hat{\theta}_j)$, which is given by \eqref{eqn: batch independent}.
    It also follows from \cite{islam2019recursive} that \eqref{eqn: batch independent A} through \eqref{eqn: recursive independent Pk} hold.
    Finally, it follows from \eqref{eqn: Jk independent column weighting} and \eqref{eqn: thetahat_j defn} that
    \begin{align*}
        \argmin_{\hat{\theta} \in \BBR^{n \times m}} J_k(\hat{\theta}) 
        & = \begin{bmatrix}
           \argmin\limits_{\hat{\theta}_1 \in \BBR^{n}}J_{1,k}(\hat{\theta}_1)  & \hspace{-1.5pt}
           \hdots & 
           \hspace{-1.5pt} \argmin\limits_{\hat{\theta}_m \in \BBR^{n}}J_{m,k}(\hat{\theta}_m) 
        \end{bmatrix}
\end{align*}
    and \eqref{eqn: independent minimizer partition} follows.
\end{proof}
}

An advantage of the column-by-column approach versus vec-permutation is that no Kronecker product is used, implying that no sparse matrices are introduced.
\blue{Note that the column-by-column approach can also be derived by applying standard least squares methods to the $m$ columns of \eqref{eqn: yk = phik theta} \cite[p. 102]{ljung1983theory}.
However, our derivation further shows the connection to vec-permutation and how the column-by-column approach implicitly implies independent weighting to the columns of the residual error and parameter regularization.}
\red{The computational complexities of BLS and RLS with independent column weighting are shown in Tables \ref{table: batch} and \ref{table: recursive} respectively.
Independent column weighting results in a $\mathcal{O}(m^2)$ times improvement in computational complexity over vec-permutation for both BLS and RLS as well as a $\mathcal{O}(m)$ times improvement in space complexity.}

\section{Matrix Update Least Squares}

Finally, we make the stronger assumption that there exists positive-definite $R \in \BBR^{n \times n}$ such that $\bar{R}$ is block diagonal of the form
\begin{align}
    \bar{R} = 
    \diag(R,\hdots,R)
    = I_m \otimes R.
    \label{eqn: Rbar block diag identical}
\end{align}
Furthermore, we assume that, for all $k \ge 0$, there exist positive-definite $\Gamma_k \in \BBR^{p \times p}$ such that $\bar{\Gamma}_k$ is block diagonal of the form
\begin{align}
    \bar{\Gamma}_k = 
    \diag(\Gamma_k,\hdots,\Gamma_k)
    = I_m \otimes \Gamma_k.
    \label{eqn: Gammabar block diag identical}
\end{align}
This corresponds to independent and identical weighting of the columns of the residual error, $y_i - \phi_i \hat{\theta}$, $i = 0,\hdots,k$, and of the regularization term, $\hat{\theta} - \theta_0$, in \eqref{eqn: RLS cost, no assumptions}. 
Then, it follows from Lemma \ref{lem: vec and trace} that, for all $k \ge 0$, \eqref{eqn: RLS cost, no assumptions} can be rewritten as
\begin{align}
\label{eqn: Jk independent identical}
    J_k(\hat{\theta}) = \tr \bigg[ & \sum_{i=0}^k (y_i - \phi_i \hat{\theta})^\rmT \Gamma_i (y_i - \phi_i \hat{\theta}) \nonumber
    \\
    & + (\theta - \theta_0)^\rmT R (\theta - \theta_0) \bigg].
\end{align}
\red{Propositions \ref{prop: batch independent identical} gives the independent identical column weighting batch and recursive least squares methods to minimize cost function $J_k$ given by \eqref{eqn: Jk independent identical}.}
\blue{Propositions \ref{prop: batch independent identical} show that, under assumptions \eqref{eqn: Rbar block diag identical} and \eqref{eqn: Gammabar block diag identical}, the cost function $J_k$, given by \eqref{eqn: Jk independent identical}, can be minimized as a single matrix equation. We call this the \textit{matrix update} approach.}

\blue{
\begin{prop}
\label{prop: batch independent identical}
    For all $k \ge 0$, let $\phi_k \in \BBR^{p \times n}$, let $y_k \in \BBR^{p \times m}$, and let ${\Gamma}_k \in \BBR^{p \times p}$ be positive definite. 
    Furthermore, let ${\theta}_0 \in \BBR^{n \times m}$  and let ${R} \in \BBR^{n \times n}$ be positive definite. 
    Then, for all $k \ge 0$, $J_k \colon \BBR^n \rightarrow \BBR$, defined in \eqref{eqn: Jk independent identical}, has a unique minimizer, denoted as $\theta_{k+1} \triangleq \argmin_{\hat{\theta} \in \BBR^{n \times m}} J_k(\hat{\theta})$, which is given by
    \begin{align}
    \label{eqn: Batch independent identical}
        \theta_{k+1}
        = A_{k}^{-1} b_{k},
    \end{align}
    where
    \begin{align}
        A_{k} &\triangleq R + \sum_{i=0}^k \phi_i^\rmT \Gamma_i\phi_i, 
        \label{eqn: Batch independent identical A}
        \\
        b_{k} &\triangleq R \theta_0 + \sum_{i=0}^k \phi_i^\rmT \Gamma_i y_i.
        \label{eqn: Batch independent identical b}
    \end{align}
    Furthermore, for all $k \ge 0$, $\theta_{k+1} \in \BBR^{n \times m}$ is given recursively by
    \begin{align}
        P_{k+1}^{-1} &= P_k^{-1} + \phi_k^\rmT \Gamma_k \phi_k, 
        \label{eqn: Recursive independent identical Pkinv}
        \\
        \theta_{k+1} &= \theta_k + P_{k+1} \phi_k^\rmT \Gamma_k (y_k - \phi_k \theta_k).
        \label{eqn: Recursive independent identical thetak}
    \end{align}
    where $P_{0} \triangleq R^{-1}$ and, for all $k \ge 0$, $P_k \in \BBR^{n \times n}$ is positive definite, hence nonsingular.
    Moreover, for all $k \ge 0$, $P_{k+1}$ can be expressed recursively as
    \begin{align}
        P_{k+1} &= P_k - P_k \phi_k^\rmT(\Gamma_k^{-1} + \phi_k P_k \phi_k^\rmT)^{-1} \phi_k P_k.
        \label{eqn: Recursive independent identical Pk}
    \end{align}
\end{prop}
}
\blue{
\begin{proof}
    Note that \eqref{eqn: Jk independent identical} can be rewritten as \eqref{eqn: Jk independent column weighting} where, for all $k \ge 0$ and $j = 1,\hdots,m$, $\Gamma_{j,k} = \Gamma_k$ and $R_j = R$. It then follows from Proposition \ref{prop: batch independent} that $J_k$ has a unique minimizer given by $\theta_{k+1} 
        = \begin{bmatrix}
            A_k^{-1} b_{1,k} & \cdots & A_k^{-1} b_{m,k}
        \end{bmatrix}
        = A_k^{-1}  \begin{bmatrix}
            b_{1,k} & \cdots & b_{m,k}
        \end{bmatrix}$,
    where $A_k \in \BBR^{n \times n}$ is defined in \eqref{eqn: Batch independent identical A} and $b_{1,k},\hdots,b_{m,k}$ are defined in \eqref{eqn: batch independent b}.
    Finally, note that $b_{k} = [b_{1,k} \ \cdots \ b_{m,k}]$, yielding \eqref{eqn: Batch independent identical}.
    Moreover, it also follows from Proposition \ref{prop: batch independent} that, for all $k \ge 0$, $P_{k+1}^{-1} = P_{k}^{-1} + \phi_{k}^\rmT \Gamma_{k} \phi_{k},$ and $\theta_{j,k+1} = \theta_{j,k} + P_{k+1} \phi_{k}^\rmT \Gamma_{k} (y_{j,k} - \phi_k \theta_{j,k})$,
    %
    %
    where, for all $j = 1,\hdots,m$, $\theta_{j,k} \in \BBR^n$ and $y_{j,k} \in \BBR^p$ are the $j^{\rm th}$ columns of $\theta_k$ and $y_k$, respectively. Combining columns yields \eqref{eqn: Recursive independent identical thetak} and applying matrix inversion lemma to \eqref{eqn: Recursive independent identical Pkinv} yields \eqref{eqn: Recursive independent identical Pk}.
\end{proof}
}
\blue{Similarly, while the matrix update approach can be derived by applying standard least squares methods to the $m$ columns of \eqref{eqn: yk = phik theta} and combining columns \cite[p. 103]{ljung1983theory},
our derivation shows how the matrix update approach implicitly implies independent and identical weighting to the columns of the residual error and parameter regularization.}
%
%
\subsection{Convergence of Matrix Update RLS}


It is well-known that in standard RLS, the parameter estimate vector converges to the vector of true parameters if the sequence of regressors $(\phi_k)_{k=0}^\infty$ is persistently exciting \cite{bruce2021necessary,lai2021regularization}.
Theorem \ref{theo::PE_implies_1/k} extends this result to matrix RLS.
\blue{To begin, we extend the definition of persistent excitation (PE) from page 64 of \cite{aastrom1995adaptive} to the case of matrix regressors and nonuniform weight.}\footnote{\blue{Note that \cite{aastrom1995adaptive} considers a measurement process $y_k = \phi_k^\rmT \theta$ while we consider $y_k = \phi_k \theta$. As such, PE is defined in \cite{aastrom1995adaptive} using the limit $\lim_{k \rightarrow \infty} \frac{1}{k} \sum_{i=0}^{k-1} \phi_i \phi_i^\rmT$ whereas we consider $\lim_{k \rightarrow \infty} \frac{1}{k} \sum_{i=0}^{k-1} \phi_i^\rmT \phi_i$.}}

\begin{defn}
$(\phi_k)_{k=0}^\infty\subset \BBR^{p\times n}$ with weight $(\Gamma_k)_{k=0}^\infty$ is persistently exciting (PE) if 
\begin{align}
    C \triangleq \lim_{k \rightarrow \infty} \frac{1}{k} \sum_{i=0}^{k-1} \phi_i^\rmT \Gamma_i \phi_i  \in \BBR^{n \times n}
    \label{eq::PE_matrix_C}
\end{align}
exists and is positive definite. 
\end{defn}

\begin{theo}
\label{theo::PE_implies_1/k}
Let $\theta, {\theta}_0 \in \BBR^{n \times m}$  and let ${R} \in \BBR^{n \times n}$ be positive definite.
For all $k \ge 0$, let $\phi_k \in \BBR^{p \times n}$, let $y_k \in \BBR^{p \times m}$ be given by \eqref{eqn: yk = phik theta}, let ${\Gamma}_k \in \BBR^{p \times p}$ be positive definite, and let $P_k \in \BBR^{n \times n}$ and $\theta_k \in \BBR^{n \times m}$ be given by \eqref{eqn: Recursive independent identical Pkinv} and \eqref{eqn: Recursive independent identical thetak}, respectively. 
Assume that $(\phi_k)_{k=0}^\infty$ with weight $(\Gamma_k)_{k=0}^\infty$ is PE, and define $C \in \BBR^{n \times n}$ by \eqref{eq::PE_matrix_C}.  Then,
\begin{align}
    \lim_{k \rightarrow \infty}k(\theta_k - \theta) = C^{-1} R (\theta_0 - \theta).
\end{align}
\end{theo}
\noindent {\it Proof.}
Note that 
\begin{align*}
    \theta_k &= (R + \sum_{i=0}^{k-1} \phi_i^\rmT \Gamma_i\phi_i)^{-1} (R \theta_0 + \sum_{i=0}^{k-1} \phi_i^\rmT \Gamma_i y_i) 
    \\
    &= (R + \sum_{i=0}^{k-1} \phi_i^\rmT \Gamma_i\phi_i)^{-1} \Big[R (\theta_0 - \theta) + (R + \sum_{i=0}^{k-1} \phi_i^\rmT \Gamma_i \phi_i) \theta \Big]
    \\
    &= (R + \sum_{i=0}^{k-1} \phi_i^\rmT \Gamma_i\phi_i)^{-1} R (\theta_0 - \theta) + \theta.
\end{align*}
Hence, it follows that 
\begin{align}
    \lim_{k \rightarrow \infty} k(\theta_k - \theta) &= \lim_{k \rightarrow \infty} (\frac{1}{k} R + \frac{1}{k} \sum_{i=0}^{k-1} \phi_i^\rmT \Gamma_i\phi_i)^{-1} R (\theta_0 - \theta) \nonumber\\
    & = C^{-1}R (\theta_0 - \theta).\tag*{\mbox{$\square$}}
\end{align}

\begin{table*}\centering
\caption{Batch least squares summary and computational complexities for $N$ measurements}
\setlength\tabcolsep{6pt} 
\begin{tabular}{@{}rllll@{}}
\toprule[1pt] 
 & Algorithm & 
Comp. Complexity (No Assumptions) & 
\begin{tabular}{@{}c@{}} Comp. Complexity \\  $(N \gg n \ge p)$ \end{tabular} &
\begin{tabular}{@{}c@{}} Comp. Complexity \\  $(N \gg p \ge n)$ \end{tabular} 
\\
\midrule
Vec-Permutation & 
\eqref{eqn: batch vec permutation}, \eqref{eqn: batch vec permutation A}, \eqref{eqn: batch vec permutation b} & 
$\mathcal{O}(\max\{N p n m^3 \max\{p,n\} ,n^3 m^3\})$ & 
$\mathcal{O}(N p n^2 m^3)$ &
$\mathcal{O}(N p^2 n m^3)$ 
\\
\blue{Column-by-Column} & 
\eqref{eqn: batch independent}, \eqref{eqn: batch independent A}, \eqref{eqn: batch independent b} & 
$\mathcal{O}(\max\{N p n m \max\{p,n\} ,n^3 m\})$ &
$\mathcal{O}(N p n^2 m)$ &
$\mathcal{O}(N p^2 n m)$ 
\\
\blue{Matrix Update} &
\eqref{eqn: Batch independent identical}, \eqref{eqn: Batch independent identical A}, \eqref{eqn: Batch independent identical b} & 
$\mathcal{O}(\max\{N p \max\{n,m\} \max\{p,n\} ,n^3\})$ &
$\mathcal{O}(N p n \max\{n,m\}) $ &
$\mathcal{O}(N p^2 \max\{n,m\}) $ 
\\
\bottomrule[1pt]
\end{tabular}
\label{table: batch}
\end{table*}

\begin{table*}\centering
\caption{Recursive least squares summary and computational/space complexities per step}
\setlength\tabcolsep{6pt} 
\begin{tabular}{@{}rlllll@{}}
\toprule[1pt] 
& \begin{tabular}{@{}c@{}}  Algorithm \\ $(n \gg p)$ \end{tabular} & 
\begin{tabular}{@{}c@{}}  Comp. Complexity \\ $(n \gg p)$ \end{tabular} &
\begin{tabular}{@{}c@{}}  Algorithm \\ $(p \ge n)$ \end{tabular} & 
\begin{tabular}{@{}c@{}}  Comp. Complexity \\ $(p \ge n)$ \end{tabular} &
\begin{tabular}{@{}c@{}}  Number of Parameters \\ in Memory \end{tabular} 

\\
\midrule
Vec-Permutation & 
\eqref{eqn: recursive vec permutation Pk update}, \eqref{eqn: recursive vec permutation thetak update}  & 
$\mathcal{O}(p n^2 m^3)$ & 
\eqref{eqn: recursive vec permutation Pkinv update}, \eqref{eqn: recursive vec permutation thetak update}  &
$\mathcal{O}(p^2 n m^3)$  & 
$n^2m^2 + nm$ 
\\
\blue{Column-by-Column} & 
\eqref{eqn: recursive independent Pk}, \eqref{eqn: recursive independent thetak}  & 
$\mathcal{O}(p n^2 m)$ & 
\eqref{eqn: recursive independent Pkinv}, \eqref{eqn: recursive independent thetak} &
$\mathcal{O}(p^2 n m)$  & 
$n^2 m + nm$
\\
\blue{Matrix Update} &
\eqref{eqn: Recursive independent identical Pk}, \eqref{eqn: Recursive independent identical thetak} & 
$\mathcal{O}(pn \max\{n,m\})$ & 
\eqref{eqn: Recursive independent identical Pkinv}, \eqref{eqn: Recursive independent identical thetak} &
$\mathcal{O}(p^2 \max\{n,m\})$  & 
$n^2 + nm$
\\
\bottomrule[1pt]
\end{tabular}
\vspace{-10pt}
\label{table: recursive}
\end{table*}

\begin{figure}
    \centering
    \includegraphics[width = .48\textwidth]{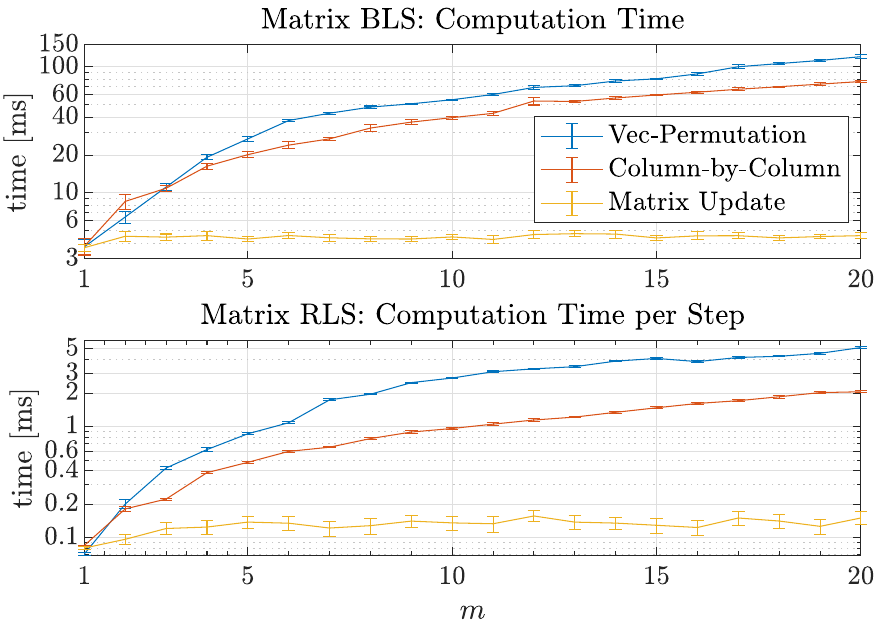}
    \caption{Consider the measurement process \eqref{eqn: yk = phik theta} with $p = 10$, $n = 50$, and $1 \le m \le 20$. Batch least squares (top) shows computation time with $N = 100$ data points, averaged over 10 trials. Recursive least squares (bottom) shows computation time per step, averaged over 100 trials. Error bars show the \qty{95}{\percent} confidence intervals.}
    \label{fig:comp_time}
\end{figure}

\section{Computational Complexity and Performance Tradeoff}
\blue{
Next, we study the computational complexities of vec-permutation, column-by-column, and matrix update least squares. For simplicity, we assume it takes $\SO(nmp)$ arithmetic operations to multiply an $(n \times m)$ matrix with a $(m \times p)$ matrix, $\SO(n^3)$ arithmetic operations to invert an $(n \times n)$ matrix, and $\SO(nm)$ arithmetic operations to add two $(n \times m)$ matrices.
}

The computational complexities of the batch and recursive least squares methods are shown in Tables \ref{table: batch} and \ref{table: recursive} respectively.
\blue{Note that for the RLS methods, columns 2 and 4 of Table \ref{table: recursive} show the most efficient implementation depending on the dimensions $n$ and $p$.}
Note that for both BLS and RLS, column-by-column and matrix update result in an $\mathcal{O}(m^2)$ and $\mathcal{O}(m^3)$ times improvement in computational complexity over vec-permutation, respectively.
Moreover, for RLS, column-by-column and matrix update result in an $\mathcal{O}(m)$ and $\mathcal{O}(m^2)$ times improvement in space complexity over vec-permutation, respectively.
Next, Figure \ref{fig:comp_time} shows numerical testing of BLS and RLS with vec-permutation, column-by-column, and matrix update. 
We consider the measurement process \eqref{eqn: yk = phik theta} with $p = 10$, $n = 50$, and $1 \le m \le 20$.
For larger values of $m$, we see significantly faster computation time for matrix update over vec-permutation and column-by-column.

\blue{While matrix update least squares offers a significant improvement in computational cost for large values of $m$, the following example shows that there may be a sacrifice in performance if there is prior knowledge that measurement noise has highly correlated columns.
While not shown in this example, there may also be performance sacrifice if there is prior knowledge that columns of the parameters are highly correlated.}

\blue{
\begin{example}
\label{example: correlated noise}
Consider the measurement process \eqref{eqn: yk = phik theta} with $p=2$, $m=2$, and $n = 100$. We consider $10$ independent trials where, for each trial, the two columns of the parameters $\theta$ are i.i.d. sampled from the Gaussian distribution $\mathcal{N}(0_{n \times 1},I_{n})$.
For each trial, for all $k \ge 0$, the two rows of the regressor, ${\phi}_k$ are i.i.d. sampled from $\mathcal{N}(0_{n \times 1},I_{n})$, and the measurement $y_k$ is given by $y_k = \phi_k \theta + v_k$, where $v_k \in \BBR^{p \times m}$ is the measurement noise and $\vecc(v_k)$ is i.i.d. sampled from $\mathcal{N}(0_{pm \times 1},\Sigma)$, where
\begin{align}
    \Sigma \triangleq \begin{bmatrix}
        \Sigma_{11} & 9.9 \,  \mathbf{1}_{2 \times 2} 
        \\
        9.9 \,  \mathbf{1}_{2 \times 2}   & \Sigma_{22}
    \end{bmatrix}, \quad
    \Sigma_{11} \triangleq \begin{bmatrix}
        1 & 0.99 \\ 0.99 & 1
    \end{bmatrix}
\end{align}
and $\Sigma_{22} \triangleq 100 \Sigma_{11}$.
In this setup, the second column of the measurement noise has higher variance than that of the first column, and any two elements of the measurement noise are highly correlated with a correlation coefficient of $0.99$.

We compare the performance of vec-permutation, column-by-column, and matrix update RLS. 
We set regularization terms as $\bar{\theta}_0 = 0_{nm \times 1}$ and $\bar{R} = I_{nm}$, $\theta_{j,0} = 0_{n \times 1}$ and $R_{j} = I_n$ $j = 1,2$, and $\theta_{0} = 0_{n \times m}$ and $R = I_n$, respectively.
For vec-permutation, we let $\bar{\Gamma}_k = \Sigma^{-1}$, giving the minimum variance estimator. 
For column-by-column, we let $\Gamma_{1,k} = \Sigma_{11}^{-1}$ and $\Gamma_{2,k} = \Sigma_{22}^{-1}$.
Finally, for matrix update, we consider the three choices $\Gamma_k = I_2$, $\Gamma_k = \Sigma_{11}^{-1}$, and $\Gamma_k = \Sigma_{22}^{-1}$.

Figure \ref{fig:error} shows $\Vert e_k \Vert_2$, the error at step $k$, defined as $\Vert \bar{\theta}_k - \bar{\theta} \Vert_2$, $\Vert [\theta_{1,k} \ \theta_{2,k}] - \theta_k \Vert_2$, and $\Vert \theta_k - \theta \Vert_2$ for vec-permutation, column-by-column, and matrix update RLS, respectively. 
The highlighted region gives the \qty{95}{\percent} confidence interval.
Note how identification performance of matrix update least squares varies considerable based on the choice of $\Gamma_k$. 
Vec-permutation performs similarly to column-by-column when $k \le n = 100$ but gives slightly better performance asymptotic once $k > n$.
\hfill{\large$\diamond$}
\end{example}
}

\begin{figure}[ht]
    \centering
    \includegraphics[width = .48\textwidth]{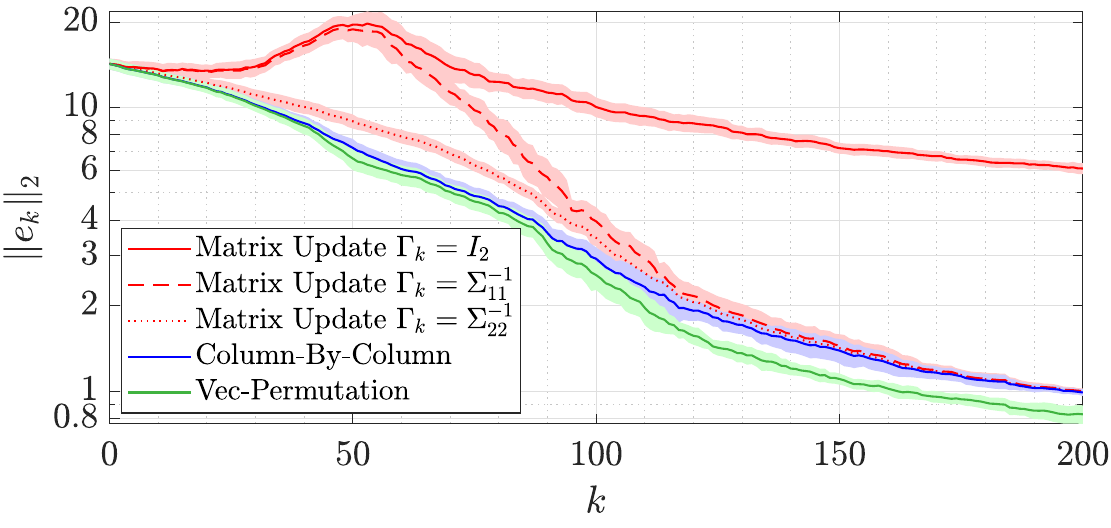}
    \caption{Example \ref{example: correlated noise}: Parameter estimation error $\Vert e_k \Vert_2$ over $200$ steps for matrix update, column-by-column, and vec-permutation recursive least squares with \qty{95}{\percent} confidence intervals over 10 trials highlighted.}
    \label{fig:error}
\end{figure}
\section{Application to Online Identification for Indirect Adaptive Model Predictive Control}
\blue{Finally, a useful application of this work is efficient online identification for adaptive model predictive control.}
Consider a MIMO input-output system of the form
\begin{align}
\label{eqn: ARMA Model}
    y_{k} = -\sum_{i=1}^{\hat{n}} F_i y_{k-i} + \sum_{i=0}^{\hat{n}} G_i u_{k-i},
\end{align}
where $k \ge 0$ is the time step, $\hat{n}$ is the model order, $u_k \in \BBR^m$ is the control, $y_k \in \BBR^p$ is the measurement, and $F_1,\hdots,F_{\hat{n}} \in \BBR^{p \times p}$ and $G_0,\hdots,G_{\hat{n}}  \in \BBR^{p \times m}$ are the system coefficient matrices to be estimated.
A model of the form \eqref{eqn: ARMA Model} is identified online in the indirect adaptive model predictive control scheme: predictive cost adaptive control (PCAC) \cite{nguyen2021predictive}.
For all $k \ge 0$, the system coefficient matrices are estimated by minimizing the cost function $J_k \colon \BBR^{p \times \hat{n}(m+p) + m} \rightarrow \BBR$, defined as
\begin{equation}
\label{eqn: PCAC cost}
    J_k(\hat{\theta}) = \sum_{i=0}^k z_i^\rmT(\hat{\theta}) z_i(\hat{\theta}) + \vecc(\hat{\theta}-\theta_0)^\rmT \bar{P}_0^{-1} \vecc(\hat{\theta}-\theta_0),
\end{equation}
where $z_k \colon \BBR^{p \times \hat{n}(m+p) + m} \rightarrow \BBR^p$ is defined
\begin{align}
    z_k(\hat{\theta}) \triangleq y_k + \sum_{i=1}^{\hat{n}} \hat{F}_i y_{k-i} - \sum_{i=0}^{\hat{n}} \hat{G}_i u_{k-i},
\end{align}
$\hat{\theta} \in \BBR^{p \times \hat{n}(m+p) + m}$ are the coefficients to be estimated, defined
\begin{align}
    \hat{\theta} \triangleq \begin{bmatrix}
        \hat{F}_1 & \cdots & \hat{F}_{\hat{n}} & \hat{G}_0 & \cdots & \hat{G}_{\hat{n}} 
    \end{bmatrix},
\end{align}
and where $\theta_0 \in \BBR^{p \times \hat{n}(m+p) + m}$ is an initial guess of the coefficients and $\bar{P}_0 \in \BBR^{[\hat{n}p(m+p) + mp] \times [\hat{n}p(m+p) + mp]}$ is positive definite. 
Note that, for all $k \ge 0$, $z_k(\hat{\theta})$ can be written as 
\begin{align}
\label{eqn: zk = y - theta phi}
    z_k(\hat{\theta}) = y_k - \hat{\theta} \phi_k,
\end{align}
where $\phi_k \in \BBR^{\hat{n}(m+p) + m}$ is defined as
\begin{align}
    \phi_k \triangleq \begin{bmatrix}
        -y_{k-1}^\rmT & \cdots -y_{k-\hat{n}}^\rmT & u_k^\rmT & \cdots & u_{k-\hat{n}}^\rmT
    \end{bmatrix}^\rmT.
\end{align}
Further defining $\bar{\phi}_k \in \BBR^{p \times \hat{n}p(m+p) + mp}$ as
\begin{align}
    \bar{\phi}_k \triangleq \phi_k^\rmT \otimes I_p,
\end{align}
it follows that $z_k(\hat{\theta})$ can be written as 
\begin{align}
\label{eqn: zk vec permutation}
    z_k(\hat{\theta}) = y_k - \bar{\phi}_k \vecc(\hat{\theta})
\end{align}
where $\vecc(\hat{\theta}) \in \BBR^{\hat{n}p(m+p) + mp}$ is the vectorization of $\hat{\theta}$. 
Using \eqref{eqn: zk vec permutation}, we derive the identification algorithm used in \cite{nguyen2021predictive}.

\begin{prop}
\label{prop: PCAC RLS}
For all $k \ge 0$, let $u_k \in \BBR^m$, $y_k \in \BBR^p$.
Furthermore, let $\theta_0 \in \BBR^{p \times \hat{n}(m+p) + m}$ and let $\bar{P}_0 \in \BBR^{[\hat{n}p(m+p) + mp] \times [\hat{n}p(m+p) + mp]}$ be positive definite. 
Then, for all $k \ge 0$, $J_k$, defined in \eqref{eqn: PCAC cost}, has a unique global minimizer, denoted
\begin{align}
    \theta_{k+1} \triangleq \argmin_{\hat{\theta} \in \BBR^{p \times \hat{n}(m+p) + m}} J_k(\hat{\theta}),
\end{align}
which is given by
\begin{align}
    \bar{P}_{k+1} &= \bar{P}_k - \bar{P}_k \bar{\phi}_k^\rmT(I_p + \bar{\phi}_k \bar{P}_k \bar{\phi}_k^\rmT)^{-1} \bar{\phi}_k \bar{P}_k, \label{eqn: Pbar update ARMA}
    \\
    \vecc(\theta_{k+1}) &= \vecc(\theta_k) + \bar{P}_{k+1} \bar{\phi}_k^\rmT (y_k - \bar{\phi}_k \vecc(\theta_k)). \label{eqn: vec theta update ARMA}
\end{align}

\end{prop}
\begin{proof}
    This result follows from Proposition \ref{prop: batch vec permutation}. For further details, see equations (8) through (20) of \cite{nguyen2021predictive}.
\end{proof}

Next, we provide an alternate formulation using matrix RLS.
\begin{prop}
Consider the notation and assumptions of Proposition \ref{prop: PCAC RLS}. If there exists $P_0 \in \BBR^{[\hat{n}(m+p) + m] \times [\hat{n}(m+p) + m]}$ such that $\bar{P_0} = P_0 \otimes I_p$, then, for all $k \ge 0$, $\theta_{k+1} \in \BBR^{p \times \hat{n}(m+p) + m}$ is given by
\begin{align}
    P_{k+1} &= P_k - \frac{P_k \phi_k \phi_k^\rmT P_k}{1 + \phi_k^\rmT P_k \phi_k},
    \label{eqn: Pk update ARMA}
    \\
    \theta_{k+1} &= \theta_k + (y_k - \theta_k \phi_k ) \phi_k^\rmT P_{k+1}.
    \label{eqn: thetak update ARMA}
\end{align}
\end{prop}
%

\begin{proof}
    This result follows from Proposition \ref{prop: batch independent identical}.
\end{proof}
\begin{example}
This example is from \cite{mohseni2022predictive} and uses PCAC for the control of a flexible structure under harmonic and broadband disturbances.
Consider the 4-bay truss show in Figure \ref{fig:Truss} made of flexible truss elements with unknown mass and stiffness. 
Two actuators are placed at nodes 3 and 4 with control authority in the $x$-direction and $x$-direction displacement sensors are placed at nodes 5, 6, 7, and 8.
The objective is to use PCAC to suppress the effects of broadband disturbances \blue{with constraints on actuator force and} without prior knowledge of the truss dynamics.
See \cite{mohseni2022predictive} for further details.

This example has inputs $u_k \in \BBR^2$ and outputs $y_k \in \BBR^4$.
Online identification was done in \cite{mohseni2022predictive} using vec-permutation, given by \eqref{eqn: Pbar update ARMA} and \eqref{eqn: vec theta update ARMA}, with identity regularization.
We replicated the results of \cite{mohseni2022predictive} using matrix RLS, given by \eqref{eqn: Pk update ARMA} and \eqref{eqn: thetak update ARMA}.
Table \ref{table: truss} shows that using matrix RLS resulted in a \qty{97.6}{\percent} decrease in computation time needed for system identification per step. 
Moreover, since system identification is a significant part of PCAC, Table \ref{table: truss} also shows a \qty{21.4}{\percent} decrease in total computation time per step. \hfill{\large$\diamond$}

\begin{figure}[ht]
    \centering
    \includegraphics[width = .15 \textwidth]{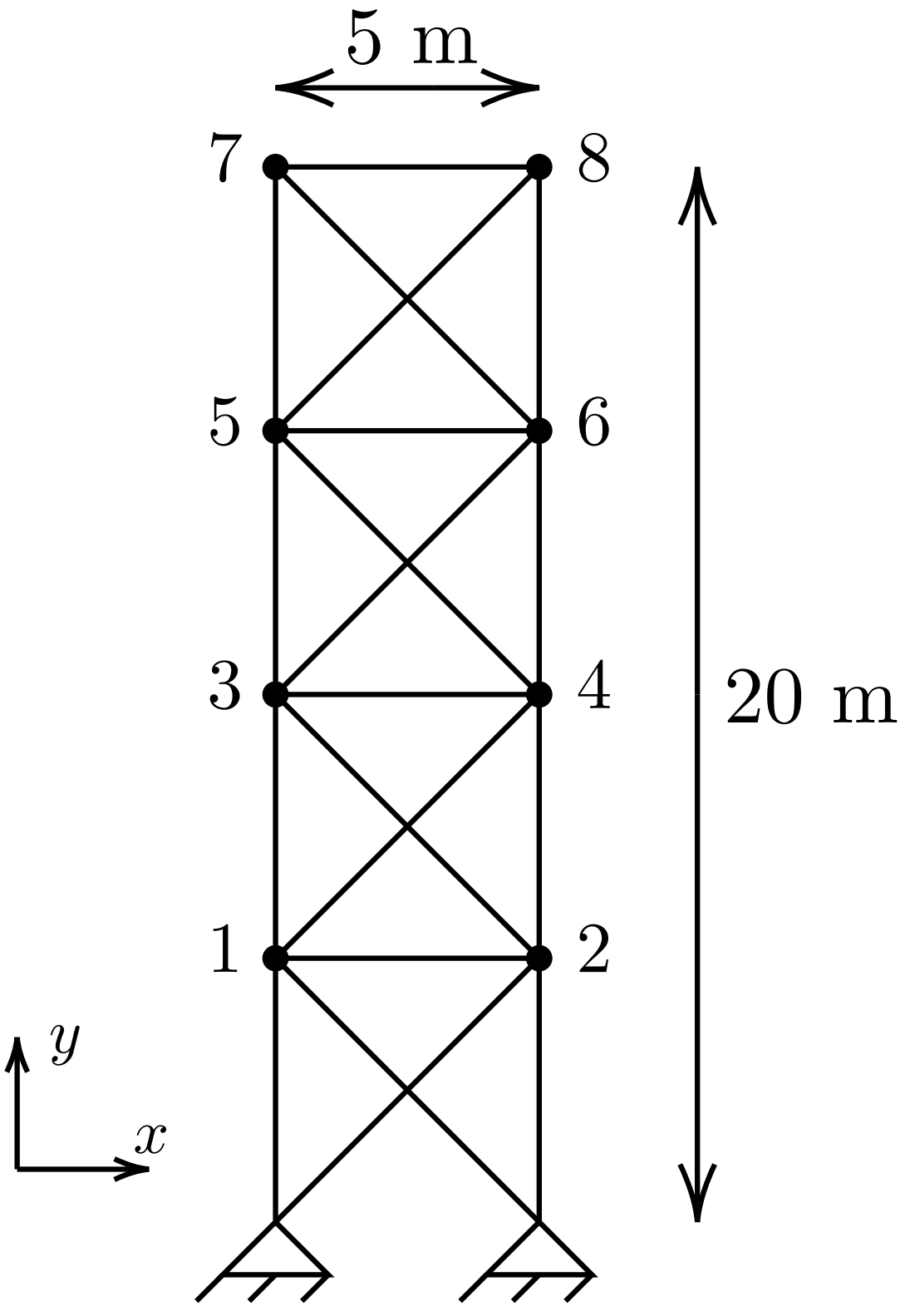}
    \caption{Flexible truss structure from \cite{mohseni2022predictive} with nodes labeled.}
    \label{fig:Truss}
\end{figure}

\begin{table}[ht]\centering
\caption{Truss ex. computation time per step: mean $\pm$ std. deviation }
\setlength\tabcolsep{6pt} 
\begin{tabular}{@{}rlll@{}}
\toprule[1pt] 
 & Vec-Permutation & 
Matrix RLS &
Change
\\
\midrule
ID Time per Step & 
\SI[separate-uncertainty = true]{5.9(5)}{\milli\second} & 
\SI[separate-uncertainty = true]{.14(3)}{\milli\second} &
\qty{-97.6}{\percent}
\\
Total Time per Step & 
\SI[separate-uncertainty = true]{28(6)}{\milli\second} & 
\SI[separate-uncertainty = true]{22(6)}{\milli\second} &
\qty{-21.4}{\percent}
\\
\bottomrule[1pt]
\end{tabular}
\vspace{-10pt}
\label{table: truss}
\end{table}

\end{example} 

\section{Conclusions}
This work derives batch and recursive least squares algorithms for the identification of matrix parameters. 
Under the assumption of independent, identical column weighting, this method minimizes the same cost function as the vec-permutation approach while significantly improving computational complexity.
It is also shown how, under persistent excitation, convergence guarantees can be extended from the vector case to the matrix case.
This approach can be used fast online identification of MIMO systems which is critical in indirect adaptive model predictive control.
\blue{A future area of interest is studying how various RLS forgetting algorithms (e.g. \cite{lai2024generalized}) can be applied to matrix update RLS.}

\bibliographystyle{IEEEtran}
\bibliography{refs}
\appendix
%

%
\begin{lema}{A.1}
\label{lem: vec and trace}
Let $x \in \BBR^{n \times m}$ and let $A \in \BBR^{n \times n}$. Then, $\vecc(x)^\rmT (I_m \otimes A) \vecc(x) = \tr(x^\rmT A x)$.
%
%
\end{lema}
%
%
%
%
\red{
\begin{lema}{A.3}
\label{lem: least squares}
    For all $k \ge 0$, let $\phi_k \in \BBR^{p \times n}$, let $y_k \in \BBR^p$, and let $\Gamma_i \in \BBR^{p \times p}$ be positive definite. Furthermore, let $\theta_0 \in \BBR^n$  and let $P_0 \in \BBR^{n \times n}$ be positive definite. 
    For all $k \ge 0$, define function $J_k \colon \BBR^n \rightarrow \BBR$ as $J_k(\hat{\theta}) = \sum_{i=0}^k (y_i - \phi_k \hat{\theta})^\rmT \Gamma_i (y_i - \phi_i \hat{\theta}) + (\hat{\theta} - \theta_0)^\rmT P_0^{-1} (\hat{\theta} - \theta_0)$.
    %
    %
    Then, for all $k \ge 0$, $J_k$ has a unique minimizer, denoted as $\theta_{k+1} \triangleq \argmin_{\hat{\theta} \in \BBR^n} J_k(\hat{\theta})$. Furthermore, for all $k \ge 0$, 
    \begin{align}
        \theta_{k+1} = A_k^{-1} b_k,
    \end{align}
    where $A_k \triangleq P_0^{-1} + \sum_{i=0}^k \phi_i^\rmT \Gamma_i \phi_i$ and $b_k \triangleq P_0^{-1} \theta_0 + \sum_{i=0}^k \phi_i^\rmT \Gamma_i y_i$.
    %
    %
    Moreover, for all $k \ge 0$, $\theta_{k+1}$ can be expressed recursively as
    \begin{align}
        P_{k+1}^{-1} &= P_k^{-1} + \phi_k^\rmT \Gamma_k \phi_k, \\
        \theta_{k+1} &= \theta_k + P_{k+1} \phi_k^\rmT \Gamma_k (y_k - \phi_k \theta_k),
    \end{align}
    where, for all $k \ge 0$, $P_k \in \BBR^{n \times n}$ is positive definite, hence nonsingular. 
    Finally, for all $k \ge 0$, $P_{k+1}$ can be expressed recursively as
    \begin{align}
        P_{k+1} &= P_k - P_k \phi_k^\rmT (\Gamma_k^{-1} + \phi_k P_k \phi_k^\rmT)^{-1} \phi_k P_k.
    \end{align}
\end{lema}
\begin{proof}
    See \cite{islam2019recursive}.
\end{proof}
}

\end{document}

%% file: shortcuts.tex
\newtheorem{lma}{Lemma}
\newtheorem{theo}{Theorem}
\newtheorem{prop}{Proposition}

\newtheorem{defn}{Definition}

\newtheorem{lem}{Lemma}

\DeclareMathOperator{\tr}{tr}
\DeclareMathOperator{\vecc}{vec}
\DeclareMathOperator{\diag}{diag}

\newcommand{\rmT}{{\rm T}}

\newcommand{\BBR}{{\mathbb R}}

\newcommand{\SO}{{\mathcal O}}

%% file: main_color_v2.bbl
\begin{thebibliography}{10}
\providecommand{\url}[1]{#1}
\csname url@samestyle\endcsname
\providecommand{\newblock}{\relax}
\providecommand{\bibinfo}[2]{#2}
\providecommand{\BIBentrySTDinterwordspacing}{\spaceskip=0pt\relax}
\providecommand{\BIBentryALTinterwordstretchfactor}{4}
\providecommand{\BIBentryALTinterwordspacing}{\spaceskip=\fontdimen2\font plus
\BIBentryALTinterwordstretchfactor\fontdimen3\font minus \fontdimen4\font\relax}
\providecommand{\BIBforeignlanguage}[2]{{%
\expandafter\ifx\csname l@#1\endcsname\relax
\typeout{** WARNING: IEEEtran.bst: No hyphenation pattern has been}%
\typeout{** loaded for the language `#1'. Using the pattern for}%
\typeout{** the default language instead.}%
\else
\language=\csname l@#1\endcsname
\fi
#2}}
\providecommand{\BIBdecl}{\relax}
\BIBdecl

\bibitem{aastrom1995adaptive}
K.~J. Astrom and B.~Wittenmark, \emph{Adaptive Control}, 2nd~ed.\hskip 1em plus 0.5em minus 0.4em\relax USA: Addison-Wesley Longman Publishing Co., Inc., 1994.

\bibitem{ljung1983theory}
L.~Ljung and T.~S{\"o}derstr{\"o}m, \emph{Theory and practice of recursive identification}.\hskip 1em plus 0.5em minus 0.4em\relax MIT press, 1983.

\bibitem{islam2019recursive}
S.~Islam and D.~S. Bernstein, ``Recursive least squares for real-time implementation,'' \emph{IEEE Ctrl. Sys. Mag.}, vol.~39, no.~3, pp. 82--85, 2019.

\bibitem{nguyen2021predictive}
T.~W. Nguyen \emph{et~al.}, ``Predictive cost adaptive control: A numerical investigation of persistency, consistency, and exigency,'' \emph{IEEE Control Systems Magazine}, vol.~41, no.~6, pp. 64--96, 2021.

\bibitem{henderson1981vec}
H.~V. Henderson and S.~R. Searle, ``The vec-permutation matrix, the vec operator and kronecker products: A review,'' \emph{Linear and multilinear algebra}, vol.~9, no.~4, pp. 271--288, 1981.

\bibitem{islam2021data}
S.~A.~U. Islam \emph{et~al.}, ``Data-driven retrospective cost adaptive control for flight control applications,'' \emph{Journal of Guidance, Control, and Dynamics}, vol.~44, no.~10, pp. 1732--1758, 2021.

\bibitem{zhu2021recursive}
K.~Zhu, C.~Yu, and Y.~Wan, ``Recursive least squares identification with variable-direction forgetting via oblique projection decomposition,'' \emph{Journal of Automatica Sinica}, vol.~9, no.~3, pp. 547--555, 2021.

\bibitem{ding2013coupled}
F.~Ding, ``Coupled-least-squares identification for multivariable systems,'' \emph{IET Control Theory \& Applications}, vol.~7, no.~1, pp. 68--79, 2013.

\bibitem{wang2018recursive}
Y.~Wang, F.~Ding, and M.~Wu, ``Recursive parameter estimation algorithm for multivariate output-error systems,'' \emph{Journal of the Franklin Institute}, vol. 355, no.~12, pp. 5163--5181, 2018.

\bibitem{mohseni2022predictive}
N.~Mohseni and D.~S. Bernstein, ``Predictive cost adaptive control of flexible structures with harmonic and broadband disturbances,'' in \emph{2022 American Control Conference (ACC)}.\hskip 1em plus 0.5em minus 0.4em\relax IEEE, 2022, pp. 3198--3203.

\bibitem{farahmandi2024predictive}
A.~Farahmandi and B.~Reitz, ``Predictive cost adaptive control of a planar missile with unmodeled aerodynamics,'' in \emph{AIAA SCITECH 2024 Forum}, 2024, p. 2218.

\bibitem{ma2019recursive}
H.~Ma, J.~Pan, L.~Lv, G.~Xu, F.~Ding, A.~Alsaedi, and T.~Hayat, ``Recursive algorithms for multivariable output-error-like arma systems,'' \emph{Mathematics}, vol.~7, no.~6, p. 558, 2019.

\bibitem{peterka1975square}
V.~Peterka, ``A square root filter for real time multivariate regression,'' \emph{Kybernetika}, vol.~11, no.~1, pp. 53--67, 1975.

\bibitem{ding2009multiinnovation}
F.~Ding, P.~X. Liu, and G.~Liu, ``Multiinnovation least-squares identification for system modeling,'' \emph{IEEE Transactions on Systems, Man, and Cybernetics, Part B (Cybernetics)}, vol.~40, no.~3, pp. 767--778, 2009.

\bibitem{bamieh2002identification}
B.~Bamieh and L.~Giarre, ``Identification of linear parameter varying models,'' \emph{International Journal of Robust and Nonlinear Control: IFAC-Affiliated Journal}, vol.~12, no.~9, pp. 841--853, 2002.

\bibitem{bruce2021necessary}
A.~L. Bruce, A.~Goel, and D.~S. Bernstein, ``Necessary and sufficient regressor conditions for the global asymptotic stability of recursive least squares,'' \emph{Systems \& Control Letters}, vol. 157, p. 105005, 2021.

\bibitem{lai2021regularization}
B.~Lai, S.~A.~U. Islam, and D.~S. Bernstein, ``Regularization-induced bias and consistency in recursive least squares,'' in \emph{2021 American Control Conference (ACC)}.\hskip 1em plus 0.5em minus 0.4em\relax IEEE, 2021, pp. 3987--3992.

\bibitem{hastie2009elements}
T.~Hastie \emph{et~al.}, \emph{The elements of statistical learning: data mining, inference, and prediction}.\hskip 1em plus 0.5em minus 0.4em\relax Springer, 2009, vol.~2.

\bibitem{lai2024generalized}
B.~Lai and D.~S. Bernstein, ``Generalized forgetting recursive least squares: Stability and robustness guarantees,'' \emph{IEEE Transactions on Automatic Control}, 2024.

\end{thebibliography}
